\documentclass[a4paper, 11pt]{article}
 
\usepackage{microtype}
\usepackage{fullpage}
\usepackage{amstext,amsmath,amsthm,amssymb,amscd}
\usepackage{paralist}
\usepackage{algorithm}
\usepackage{algpseudocode}
\usepackage{xspace}
\usepackage[UKenglish]{babel}
\usepackage{authblk}
\usepackage{nccmath}
\usepackage{bm}
\usepackage{hyperref}
\usepackage{float}

\title{Smaller parameters for vertex cover kernelization}
\author{Eva-Maria C. Hols}
\author{Stefan Kratsch}
\affil{Department of Computer Science, University of Bonn, Germany \footnote{Email adresses: \{hols,kratsch\}@informatik.hu-berlin.de} }

\theoremstyle{plain}
\newtheorem{proposition}{Proposition}

\newtheorem{theorem}{Theorem}
\newtheorem{definition}{Definition}
\newtheorem{lemma}{Lemma}
\newtheorem{corollary}{Corollary}
\theoremstyle{remark}

\newcommand{\prob}[1]{\textsc{\lowercase{#1}}}

\newcommand{\C}{\mathcal{C}\xspace}
\newcommand{\Oh}{\mathcal{O}}
\newcommand{\IS}{\prob{Independent Set}\xspace}
\newcommand{\VC}{\prob{Vertex Cover}\xspace}
\newcommand{\FVS}{\prob{Feedback Vertex Set}\xspace}
\newcommand{\ISLP}{\mathrm{LP_{IS}}}
\newcommand{\VCLP}{\mathrm{LP_{VC}}}
\newcommand{\Conf}[2]{\textsc{Conf}_{{#1}}({#2})}

\newcommand{\NP}{\ensuremath{\mathsf{NP}}\xspace}

\newcommand{\ncontainment}{\ensuremath{\mathsf{NP\nsubseteq coNP/poly}}\xspace}
\newcommand{\containment}{\ensuremath{\mathsf{NP\subseteq coNP/poly}}\xspace}

\makeatletter
\def\namedlabel#1#2{\begingroup
    #2%
    \def\@currentlabel{#2}%
    \phantomsection\label{#1}\endgroup
}
\makeatother

\begin{document}

\maketitle

\begin{abstract}
We revisit the topic of polynomial kernels for \VC relative to structural parameters. Our starting point is a recent paper due to Fomin and Str{\o}mme [WG 2016] who gave a kernel with $\Oh(|X|^{12})$ vertices when $X$ is a vertex set such that each connected component of $G-X$ contains at most one cycle, i.e., $X$ is a modulator to a pseudoforest. We strongly generalize this result by using modulators to $d$-quasi-forests, i.e., graphs where each connected component has a feedback vertex set of size at most $d$, and obtain kernels with $\Oh(|X|^{3d+9})$ vertices. Our result relies on proving that minimal blocking sets in a $d$-quasi-forest have size at most $d+2$. This bound is tight and there is a related lower bound of $\Oh(|X|^{d+2-\epsilon})$ on the bit size of kernels.

In fact, we also get bounds for minimal blocking sets of more general graph classes: For $d$-quasi-bipartite graphs, where each connected component can be made bipartite by deleting at most $d$ vertices, we get the same tight bound of $d+2$ vertices. For graphs whose connected components each have a vertex cover of cost at most $d$ more than the best fractional vertex cover, which we call $d$-quasi-integral, we show that minimal blocking sets have size at most $2d+2$, which is also tight. Combined with existing randomized polynomial kernelizations this leads to randomized polynomial kernelizations for modulators to $d$-quasi-bipartite and $d$-quasi-integral graphs. There are lower bounds of $\Oh(|X|^{d+2-\epsilon})$ and $\Oh(|X|^{2d+2-\epsilon})$ for the bit size of such kernels.
\end{abstract}

\section{Introduction}

The \VC problem plays a central role in parameterized complexity. In particular, it has been very important for the development of new kernelization techniques and the study of structural parameters. As a result of this work, there is now a solid understanding of which parameterizations of \VC lead to fixed-parameter tractability or existence of a polynomial kernelization. This is motivated by the fact that parameterization by solution size leads to large parameter values on many types of easy instances. Thus, while there is a well-known kernelization for instances of \VC{}($k$) to at most $2k$ vertices, it may be more suitable to apply a kernelization with a size guarantee that is a larger function but depends on a smaller parameter.

Jansen and Bodlaender~\cite{JansenB13} were the first to study kernelization for \VC under different, smaller parameters. Their main result is a polynomial kernelization to instances with $\Oh(|X|^3)$ vertices when $X$ is a feedback vertex set of the input graph, also called a \emph{modulator} to the class of forests. Clearly, the size of $X$ is a lower bound on the vertex cover size (as any vertex cover is a modulator to an independent set). Since then, their result has been generalized and complemented in several ways. The two main directions of follow-up work are to use modulators to other tractable cases instead of forests (see below) and parameterization above lower bounds (see related work). 

For any graph class $\C$, we can define a parameterization of \VC by distance to $\C$, i.e., by the minimum size of a modulator $X$ such that $G-X$ belongs to $\C$. For fixed-parameter tractability and kernelization of the arising parameterized problem it is necessary that \VC is tractable on inputs from $\C$. For hereditary classes $\C$, this condition is also sufficient for fixed-parameter tractability but not necessarily for the existence of a polynomial kernelization. Interesting choices for $\C$ are various well-studied hereditary graph classes, like forests, bipartite, or chordal graphs, and graphs of bounded treewidth, bounded treedepth, or bounded degree.

Majumdar et al.~\cite{MajumdarRS15} studied \VC parameterized by (the size of) a modulator $X$ to a graph of maximum degree at most $d$. For $d\geq 3$ this problem is \NP-hard but for $d=2$ and $d=1$ they obtained kernels with $\Oh(|X|^5)$ and $\Oh(|X|^2)$ vertices, respectively. Their result motivated Fomin and Str{\o}mme~\cite{FominS16} to investigate a parameter that is smaller than both a modulator to degree at most two and the size of a feedback vertex set: They consider $X$ being a modulator to a \emph{pseudoforest}, i.e., with each connected component of $G-X$ having at most one cycle. For this they obtain a kernelization to $\Oh(|X|^{12})$ vertices, generalizing (except for the size) the results of Majumdar et al.~\cite{MajumdarRS15} and Jansen and Bodlaender~\cite{JansenB13}. They also prove that the parameterization by a modulator to so-called \emph{mock forests}, where no cycles share a vertex, admits no polynomial kernelization unless \containment (and the polynomial hierarchy collapses).

For their kernelization, Fomin and Str{\o}mme~\cite{FominS16} prove that minimal blocking sets in a pseudoforest have size at most three, which requires a lengthy proof. (A minimal blocking set is a set of vertices whose deletion decreases the independence number by exactly one.)\footnote{Like previous work~\cite{JansenB13,FominS16} we prefer to work with \IS rather than \VC, but this makes no important difference.} This allows to reduce the number of components of the pseudoforest such that one can extend the modulator $X$ to a sufficiently small feedback vertex set by adding one (cycle) vertex per component to $X$. At this point, the kernelization of Jansen and Bodlaender~\cite{JansenB13} can be applied to get the result.

The results of Fomin and Str\o{}mme~\cite{FominS16} suggest that the border for existence of polynomial kernels for feedback vertex set-like parameters may be much more interesting than expected previously. Arguably, there is still quite some room between allowing a single cycle per component and allowing an arbitrary number of cycles so long as they share no vertices. Do larger numbers of cycles per component still allow a polynomial kernelization? Similarly, cycles in the lower bound proof have odd length and it is known that absence of odd cycles is sufficient, i.e., a kernelization for modulators to bipartite graphs is known. Could this be extended to allowing bipartite graphs with one or more odd cycles per connected component?

\subparagraph{Our work.}
We show that the answers to the above questions are largely positive and provide, essentially, a single elegant proof to cover them. To this end, it is convenient to take the perspective of feedback sets rather than the maximum size of a cycle packing. Say that a \emph{$d$-quasi-forest} is a graph such that each connected component has a feedback vertex set of size at most $d$, whereas in a \emph{$d$-quasi-bipartite graph} each connected component must have an odd cycle transversal (a feedback set for odd cycles) of size at most $d$.

We show that \VC admits a kernelization with $\Oh(|X|^{3d+9})$ vertices when $X$ is a modulator to a $d$-quasi-forest (Section \ref{section::d_forest}). The case for $d=1$ strengthens the result of Fomin and Str{\o}mme~\cite{FominS16} (as one cycle per component is stricter than feedback vertex set size one). For every fixed larger value of $d$ we obtain a polynomial kernelization, though of increasing size. The result is obtained by proving that minimal blocking sets in a $d$-quasi-forest have size at most $d+2$ (and then applying~\cite{JansenB13}). Intuitively, having a large minimal blocking set implies getting a fairly small maximum independent set because there are optimal independent sets that avoid all but any chosen vertex of a minimal blocking set. In contrast, a $d$-quasi-forest always has a large independent set because each connected component is almost a tree.

The value $d+2$ is tight already for cliques of size $d+2$, which are permissible connected components in a $d$-quasi-forest. Such cliques also imply that our parameterization inherits a lower bound of $\Oh(|X|^{d+2-\epsilon})$ from the lower bound of $\Oh(|X'|^{r-\epsilon})$ (assuming \ncontainment) for $X'$ being a modulator to a cluster graph with component size at most $r$~\cite{MajumdarRS15}.

It turns out that our proof directly extends also to $d$-quasi-bipartite graphs, proving that their minimal blocking sets similarly have size at most $d+2$ (Section \ref{section::small_blocking_sets}). Thus, when given a modulator $X$ such that $G-X$ is $d$-quasi-bipartite, we can extend it to an odd cycle transversal $X'$ of size at most $d \cdot |X|^{d+3} + |X|$, which directly yields a randomized polynomial kernel by using a randomized polynomial kernelization for \VC parameterized by an odd cycle transversal~\cite{KratschW12}.

Motivated by this, we explore also modulators to graphs in which each connected component has vertex cover size at most $d$ plus the size of a minimum fractional vertex cover, which we call $d$-quasi-integral (Section \ref{section::small_blocking_sets}). This is stronger than the previous parameter because it allows connected components that have an odd cycle transversal of size at most $d$. We show that minimal blocking sets in any $d$-quasi-integral graph have size at most $2d+2$. This bound is tight, as witnessed by the cliques with $2d+2$ vertices, and the problem inherits a lower bound of $\Oh(|X|^{2d+2-\epsilon})$ from the lower bound for modulators to cluster graphs with clique size at most $r=2d+2$~\cite{MajumdarRS15}. Using the upper bound of $2d+2$ one can remove redundant connected components until the obtained instance has vertex cover size at most $d \cdot |X|^{2d+3} + |X|$ more than the best fractional vertex cover. In other words, one can reduce to an instance of \VC parameterized above LP with parameter value $d \cdot |X|^{2d+3} + |X|$ and apply the randomized polynomial kernelization of Kratsch and Wahlstr\"om~\cite{KratschW12} to get a randomized polynomial kernel.

\subparagraph{Related work.}
Recent work of Bougeret and Sau~\cite{BougeretS16} shows that \VC admits a kernel of size $\Oh(|X|^{f(c)})$ when $X$ is a modulator to a graph of treedepth at most $c$. Their result is incomparable to ours: Already the kernelization by feedback vertex set size~\cite{JansenB13}, which we generalize, allows arbitrarily long paths in $G-X$; such paths are forbidden in a graph of bounded treedepth. Conversely, taking a star with $d$ leaves and appending a $3$-cycle at each leaf yields a graph with feedback vertex set and odd cycle transversal size equal to $d$ but constant treedepth; $d$ can be chosen arbitrarily large.

The fact that deciding whether a graph $G$ has a vertex cover of size at most $k$ is trivial when $k$ is lower than the size $MM(G)$ of a largest matching in $G$ has motivated the study of above lower bound parameters like $\ell=k-MM(G)$. The strongest lower bound employed so far is $2LP(G)-MM(G)$, where $LP(G)$ denotes the minimum cost of a fractional vertex cover, and Garg and Philip~\cite{GargP16} gave an $\Oh^*(3^{k-(2LP(G)-MM(G))})$ time algorithm. Randomized polynomial kernels are known for parameters $k-MM(G)$ and $k-LP(G)$~\cite{KratschW12} and for parameter $k-(2LP(G)-MM(G))$~\cite{Kratsch16}. Our present kernelizations are not covered even by the strongest parameter $k-(2LP(G)-MM(G))$ because already $d$-quasi-forests for any $d\geq 2$ can have a vertex cover size that is arbitrarily larger than $k-(2LP(G)-MM(G))$: Consider, for example, a disjoint union of cliques $K_4$ with four vertices each, where $2LP(K_4)-MM(K_4)=2$ but vertex cover size is three per component.

Regarding lower bounds for kernelization (all assuming \ncontainment), it is of course well known that there are no polynomial kernels for \VC when parameterized by width parameters like treewidth, pathwidth, or treedepth (cf.~\cite{BodlaenderDFH09}). Lower bounds similar to the one for modulators to mock forests by Fomin and Str{\o}mme~\cite{FominS16} were already obtained by Cygan et al.~\cite{CyganLPPS14} (modulators to treewidth at most two) and Jansen~\cite{Jansen13thesis} (modulators to outerplanar graphs). Bodlaender et al.~\cite{BodlaenderJK14} showed that there is no polynomial kernelization in terms of the vertex deletion distance to a single clique, which is stronger than distance to cluster or perfect graphs for example. Majumdar et al.~\cite{MajumdarRS15} ruled out kernels of size $\Oh(|X|^{r-\epsilon})$ when $X$ is a modulator to a cluster graph with cliques of size bounded by $r$.

\section{Preliminaries and notation} \label{section::notation}

\subparagraph{Graphs.} We use standard notation mostly following Diestel \cite{Diestel12_book}. Let $G=(V,E)$ be a graph. For a set $X \subseteq V$, let $N_G(X)$ denote the neighborhood of $X$ in $G$, i.e., $N_G(X)=\{ v \in V \setminus X \mid \exists u \in X \colon \{u,v\} \in E\}$ and let $N_G[X]$ denote the neighborhood of $X$ in $G$ including $X$, i.e., $N_G[X]=N_G(X) \cup X$. We omit the subscript whenever the underlying graph is clear from the context. Furthermore, we use $G-X$ as shorthand for $G[V \setminus X]$. For a graph $G$ we denote by $\mathrm{vc}(G)$ the vertex cover number of $G$ and by $\alpha(G)$ the independence number of $G$. Let $Y \subseteq V$, we call $Y$ a \emph{blocking set} of $G$, if deleting the vertex set $Y$ from the graph $G$ decreases the size of a maximum independent set, hence if $\alpha(G) > \alpha(G-Y)$. A blocking set $Y$ is \emph{minimal}, if no proper subset $Y'\subsetneq Y$ of $Y$ is a blocking set of $G$. We denote by $K_n$ the clique of size $n$.

\subparagraph{Linear Programming.} We denote the linear program relaxation for \VC resp.\ \IS for a graph $G=(V,E)$ by $\VCLP(G)$ resp.\ $\ISLP(G)$. Recall that $\VCLP(G)=\min\{ \sum_{v \in V} x_v \mid \forall \{u,v\} \in E \colon x_u + x_v \geq 1 \wedge \forall v\in V \colon 0 \leq x_v \leq 1\}$ and $\ISLP(G)=\max\{\sum_{v \in V} x_v \mid \forall \{u,v\} \in E \colon x_u+x_v \leq 1 \wedge \forall v \in V \colon 0\leq x_v \leq 1 \}$. A \emph{feasible solution} to one of the above linear program relaxations is an assignment to the variables $x_v$ for all vertices $v \in V$ which satisfies the conditions of the linear program. An optimum solution to $\VCLP(G)$ resp.\ $\ISLP(G)$ is a feasible solution $x$ which minimizes resp.\ maximizes the objective function value $w(x):=\sum_{v \in V} x_v$. It follows directly from the definition that $x$ is a feasible solution to $\VCLP(G)$ if and only if $x'=1-x$ is a feasible solution to $\ISLP(G)$; thus $w(x') = |V|-w(x)$. It is well known that there exists an optimum feasible solution $x$ to $\VCLP(G)$ with $x_v \in \{0,\frac{1}{2},1\}$; we call such a solution \emph{half integral}. The same is, of course, true for $\ISLP(G)$. Given a half integral solution $x$ (to $\VCLP(G)$ or $\ISLP(G)$), we define $V_i^x = \{v \in V \mid x_v=i\}$ for each $i \in \{0,\frac{1}{2},1\}$.
Note that if $x$ is an optimum half integral solution to $\VCLP(G)$, then it holds that $N(V_0^x)=V_1^x$, whereas, it holds that $N(V_1^x)=V_0^x$, when $x$ is an optimum half integral solution to $\ISLP(G)$. We omit the subscript $x$, when the solution $x$ is clear from the context.

\section{Vertex Cover parameterized by a modulator to a $d$-quasi-forest} \label{section::d_forest}

In this section we present a polynomial kernel for \VC parameterized by a modulator to a $d$-quasi-forest. More precisely, we develop a polynomial kernel for \IS parameterized by a modulator to a $d$-quasi-forest which, by the relation between these two problems, directly yields a polynomial kernel for \VC parameterized by a modulator to a $d$-quasi-forest. 

Consider an instance $(G,X,k)$ of the problem, which asks whether graph $G$, with $G-X$ is a $d$-quasi-forest, has an independent set of size $k$. 
Like Fomin and Str{\o}mme~\cite{FominS16}, we reduce the input instance $(G,X,k)$ until the $d$-quasi-forest $G-X$ has at most polynomially many connected components in terms of $|X|$; see Reduction Rule \ref{rule::delete_cc}. By adding for each component of the $d$-quasi-forest a feedback vertex set of size $d$ to the modulator $X$, we polynomially increase the size of the modulator $X$. The resulting modulator is a feedback vertex set, hence we can apply the polynomial kernelization for \IS parameterized by a modulator to a feedback vertex set from Jansen and Bodlaender~\cite{JansenB13}.

Let $(G,X,k)$ be an instance of \IS parameterized by a modulator to a $d$-quasi-forest. Since $d$ is a constant we can compute in polynomial time a maximum independent set in $G-X$. 
Choosing some vertices from the set $X$ to be in an independent set will prevent some vertices in $G-X$ to be part of the same independent set; thus it may be that we can add less than $\alpha(G-X)$ vertices from $G-X$ to an independent set that contains some vertices of $X$. To measure this difference, we use the term of conflicts introduced by Jansen and Bodlaender~\cite{JansenB13}. Our definition is more general in order to use it also for modulators to $d$-quasi-bipartite resp.\ $d$-quasi-integral graphs.

\begin{definition}[Conflicts]
    Let $G=(V,E)$ be a graph and $X \subseteq V$ be a subset of $V$, such that we can compute a maximum independent set in $G-X$ in polynomial time. Let $F$ be a subgraph of $G-X$ and let $X' \subseteq X$.
    We define the number of conflicts on $F$ which are induced by $X'$ as $\Conf{F}{X'}: =\alpha(F)-\alpha(F-N(X'))$.
\end{definition}

Now we can state our reduction rule, which deletes some connected components of the $d$-quasi-forest $G-X$. More precisely, we delete connected components $H$ of which we know that there exists a maximum independent set in $G$ that contains a maximum independent set of the connected component $H$.

\begin{description}
    \item[Reduction Rule \namedlabel{rule::delete_cc}{1}:] If there exists a connected component $H$ of $G-X$ such that for all independent sets $X_I \subseteq X$ of size at most $d+2$ with $\Conf{H}{X_I}>0$ it holds that $\Conf{G-H-X}{X_I} \geq |X|$, then delete $H$ from $G$ and reduce $k$ by $\alpha(H)$.
\end{description}

The proof of safeness will be given in the sequel.
In particular, we delete connected components that have no conflicts.
The goal of Reduction Rule \ref{rule::delete_cc} is to delete connected components of the $d$-quasi-forest $G-X$ such that we can bound the number of connected components by a polynomial in the size of $X$. Thus, if we cannot apply this reduction rule any more we should be able to find a good bound for the number of connected components in the $d$-quasi-forest $G-X$. The following lemma yields such a bound.

\begin{lemma} \label{lemma::bound_cc}
    Let $(G,X,k)$ be an instance of \IS parameterized by a modulator to a $d$-quasi-forest where Reduction Rule \ref{rule::delete_cc} is not applicable. Then the number of connected components in $G-X$ is at most $|X|^{d+3}$.
\end{lemma}

\begin{proof}
    Let $H$ be a connected component of the $d$-quasi-forest $G-X$. Since Reduction Rule \ref{rule::delete_cc} is not applicable, there exists an independent set $X_I \subseteq X$ of size at most $d+2$ such that $\Conf{H}{X_I} > 0$ and $\Conf{G-H-X}{X_I} < |X|$; otherwise Reduction Rule \ref{rule::delete_cc} would delete $H$ (or another connected component with the same properties).
    
    Observe, that there are at most $|X|$ connected components of the $d$-quasi-forest $G-X$ that have a conflict with an independent set $X_I \subseteq X$, when $X_I$ is the reason that we cannot apply Reduction Rule \ref{rule::delete_cc} to one of these connected components: Assume for contradiction that there are $p > |X|$ connected components $H_1, H_2, \ldots, H_p$ of the $d$-quasi-forest $G-X$ that have a conflict with the same independent set $X_I \subseteq X$ of size at most $d+2$; therefore it holds that $\Conf{H_i}{X_I} > 0$ for all $i \in \{1,2,\ldots, p\}$. But now, for all $i \in \{1,2,\ldots, p\}$
    \begin{align*}
        \Conf{G-H_i-X}{X_I} \geq \sum_{\overset{j=1}{j \neq i}}^p \Conf{H_j}{X_I} \geq p-1 \geq |X| \text{,} 
    \end{align*}
    where the first inequality corresponds to summing over some connected components of $G-H_i-X$.
    Thus, $X_I$ could not be the reason why the connected components $H_1, H_2, \ldots, H_p$ are not reduced during Reduction Rule \ref{rule::delete_cc}.
         
    This leads to the claimed bound of at most $\binom{|X|}{\leq d+2} \cdot |X| \leq |X|^{d+3}$ connected components in $G-X$, because for every independent set $X_I \subseteq X$ of size at most $d+2$ there are at most $|X|$ connected components for which $X_I$ is the reason that we cannot apply Reduction Rule~\ref{rule::delete_cc}. 
\end{proof}

It remains to show that Reduction Rule \ref{rule::delete_cc} is safe; i.e.\ that there exists a solution for $(G,X,k)$ if and only if there exists a solution for $(G',X,k')$, where $G'=G-H$, $k'=k-\alpha(H)$ and $H$ is the connected component of $G-X$ we delete during Reduction Rule \ref{rule::delete_cc}.
The main ingredient for this is to prove that any minimal blocking set has size at most $d+2$ (Lemma~\ref{lemma::conflict_size}).
To bound the size of minimal blocking sets we need the existence of a half integral solution $x$ to $\ISLP(G-Y)$ for which \emph{every} maximum independent set $I$ in $G-Y$ fulfills $V_1 \subseteq I \subseteq V_\frac{1}{2} \cup V_1$. This is similar to the result of Nemhauser and Trotter~\cite{NemhauserT75} and other results about the connection between maximum independent sets (resp.\ minimum vertex covers) and their fractional $\mathrm{LP}$ solutions~\cite{Abu-KhzamFLS07, boros2002, chlebik2008crown, hammer1982}.

\begin{lemma} \label{lemma::LP_solution}
    Let $G=(V,E)$ be an undirected graph. 
    There exists an optimum half integral solution $x \in \{0,\frac{1}{2},1\}^{|V|}$ to $\ISLP(G)$ such that for all maximum independent sets $I$ in $G$ it holds that $V^x_1 \subseteq I \subseteq V \setminus V^x_0$.
\end{lemma}

\begin{proof}
    Let $x \in \{0,\frac{1}{2},1\}^{|V|}$ be an optimum half integral solution to $\ISLP(G)$, such that $V^x_\frac{1}{2}$ is maximal; this means, that there exists no optimum half integral solution $x' \neq x$ to $\ISLP(G)$ such that $V_\frac{1}{2}^{x} \subsetneq V_\frac{1}{2}^{x'}$. We will show that every independent set $I$ in $G$ with $V^x_1 \nsubseteq I$ or $V^x_0 \cap I \neq \emptyset$ is not a maximum independent set in $G$.
    
    First, we observe that for all subsets $V_0' \subseteq V^x_0$ it must hold that the size of the neighborhood of $V_0'$ in $V^x_1$ is larger than the size of $V_0'$, i.e.\ $|V^x_1 \cap N(V_0')|>|V_0'|$; if this is not the case, then we can construct an optimum half integral solution $x'$ to $\ISLP(G)$ with $V^x_\frac{1}{2} \subsetneq V^{x'}_\frac{1}{2}$ (which contradicts the fact that $V_\frac{1}{2}^x$ is maximal), by assigning a value of $\frac{1}{2}$ to all vertices in $(V^x_1 \cap N(V_0')) \cup V_0'$.
    Obviously, it holds that $V_\frac{1}{2}^x \subsetneq V_\frac{1}{2}^{x'}$ and that 
    \begin{align*}
        w(x')=w(x)-|V^x_1 \cap N(V_0')| + \frac{1}{2} (|V^x_1 \cap N(V_0')|+|V_0'|) \geq w(x) \text{.}
    \end{align*}
    In order to show that $x'$ is indeed a feasible solution to $\ISLP(G)$, it suffices to consider edges $\{u,v\}$ of $G$ that have at least one endpoint in $V_0'$, say $v\in V_0'$, because these are the only vertices for which we increase the value of the half integral solution $x$ to obtain $x'$. Since $x'_v=\frac12$, the constraint $x'_u+x'_v\leq 1$ can only be violated if $x'_u = 1$. But then $x_u=1$ must hold since the only changed values are $\frac12$ in $x'$. This of course means that $u\in V_1^x\cap N(V_0')$ and $x'_u=\frac12$; a contradiction.
    
    Now, we assume that there exists a maximum independent set $I$ that contains a vertex of the set $V^x_0$. Let $V_0'=V^x_0 \cap I \neq \emptyset$. We will show that deleting the set $V_0'$ from the independent set $I$ and adding the set $N(V_0') \cap V^x_1$ to the independent set $I$ leads to a larger independent $I'$ of $G$, i.e.\ $I'=I \setminus V_0' \cup (N(V_0')\cap V^x_1)$. 
    First we show that $I'$ has larger cardinality than $I$. Since $I$ is an independent set, we know that $(N(V_0')\cap V^x_1) \cap I = \emptyset$ and hence that the cardinality of $I'$ is $|I|-|V_0'|+|N(V_0')\cap V^x_1|$. From the above observation, we know that $|N(V_0')\cap V^x_1|> |V_0'|$ and it follows that $I'$ has larger cardinality than $I$. To prove that $I'$ is an independent set in $G$, it is enough to show that any vertex $v \in N(V_0')\cap V^x_1$ has no neighbor in $I'$; this holds because $V^x_1$ is an independent set, $N(V^x_1) \subseteq V^x_0$ and $V^x_0 \cap I'= \emptyset$. Thus, $I'$ is an independent set which has larger cardinality than $I$; this contradicts the assumption that $I$ is a maximum independent set.
    
    It remains to show that there exists no maximum independent set $I$ in $G$ with $V^x_1 \nsubseteq I \subseteq V^x_1 \cup V^x_\frac{1}{2}$. Let $v \in V^x_1 \setminus I$. Since $I$ is a maximum independent set, there exists a vertex $w \in N(V^x_1) \cap I$ (otherwise $I \cup \{v\}$ would be a larger independent set in $G$). But $N(V^x_1) \subseteq V^x_0$ and hence $w \in V^x_0 \cap I$, which contradicts the assumption that $I \subseteq V^x_1 \cup V^x_\frac{1}{2}$.
\end{proof}

Using the above lemma, we can show that every minimal blocking set in a $d$-quasi-forest has size at most $d+2$. This generalizes the result of Fomin and Str{o}mme \cite{FominS16}, who showed that a minimal blocking set in a pseudoforest has size at most three. Furthermore, we can show that this bound is tight. 

\begin{theorem} \label{theorem::lower_bound_FVS}
    Minimal blocking sets have a tight upper bound of $d+2$ in $d$-quasi-forests.
\end{theorem}

The crucial part of Theorem \ref{theorem::lower_bound_FVS} is to prove the upper bound. 

\begin{lemma} \label{lemma::conflict}
    Let $H=(V,E)$ be a $d$-quasi-forest and let $Z$ be a feedback vertex set in $H$ of size at most $d$. Then it holds that a minimal blocking set $Y$ in the $d$-quasi-forest $H$ has size at most $|Z|+2 \leq d+2$.
\end{lemma}

\begin{proof}
    We consider an optimum half integral solution $x$ to $\ISLP(H-Y)$ which fulfills the properties of Lemma \ref{lemma::LP_solution}; let $V_i = \{ v \in V(H-Y) \mid x_v = i\}$ for $i\in \{0, \frac{1}{2},1\}$. We know that every maximum independent set $I$ of $H-Y$ contains the set $V_1$ and no vertex of the set $V_0$ (because $x$ fulfills the properties of Lemma \ref{lemma::LP_solution}).
    
    Observe that for all vertices $y \in Y$ it holds that $\alpha(H - (Y \setminus \{y\})) = \alpha(H)$; otherwise, the set $Y$ would not be a minimal blocking set.
    Furthermore, from the above observation it follows that $\alpha(H)= \alpha(H-Y) + 1$, because 
    \begin{align*}
        \alpha(H-Y) < \alpha(H) = \alpha (H- (Y \setminus \{y\})) \leq \alpha(H-Y) + 1 \text{ for all } y \in Y \text{.}
    \end{align*}    
    The key observation of our proof is that $N_H(Y) \subseteq V_0 \cup V_\frac{1}{2}$; this follows from the fact that $Y$ is minimal:
    As observed above, we know that $\alpha(H - (Y \setminus \{y\})) = \alpha(H)$. Thus, for all vertices $y \in Y$ there exists a maximum independent set $I_y$ in $H$ that contains the vertex $y$ and no other vertex from the set $Y$. Consider the sets $I'_y = I_y \setminus \{y\}$ for all vertices $y \in Y$. Obviously, the sets $I'_y$ are independent sets in $H-Y$ for all vertices $y \in Y$, because $y \in Y$ is the only vertex of the set $Y$ that is contained in $I_y$. Furthermore, we know that the sets $I_y'$ are maximum independent sets in $H-Y$ because
    \begin{align*}
     |I'_y|+1 = |I_y| = \alpha(H) = \alpha(H-Y)+1.
    \end{align*}
    
    The fact that $I_y'$ is a maximum independent set for all vertices $y \in Y$ implies that $V_1 \subseteq I_y' = I_y \setminus \{y\} \subseteq I_y$ (by the choice of the solution $x$ to $\ISLP(H-Y)$). Thus, for all vertices $y \in Y$ it holds that $V_1 \subseteq I_y$ and therefore that $N_H(I_y) \cap V_1 = \emptyset$ which implies that $N_H(\{y\}) \cap V_1 = \emptyset$ (because $V_1 \cup \{y\} \subseteq I_y$). Since this holds for all vertices $y \in Y$ it follows that $N_H(Y) \cap V_1 = \emptyset$, hence $N_H(Y) \subseteq V_0 \cup V_\frac{1}{2}$.
    
    To bound the size of $Y$ we try to find an upper bound for the size of a maximum independent set in $H-Y$ and a lower bound for the size of a maximum independent set in $H$.
    An obvious upper bound for the size of a maximum independent set in $H-Y$ is the optimum value of $\ISLP(H-Y)$ which is equal to $|V_1| + \frac{1}{2} |V_\frac{1}{2}|$. This leads to an upper bound for $\alpha(H-Y)$:
    \begin{align}
        \alpha(H-Y) &\leq w(x) = |V_1| + \frac{1}{2} |V_\frac{1}{2}| = |V_1| + \frac{1}{2}|H-V_0-V_1-Y| \nonumber \\ &= |V_1| + \frac{|H-V_0-V_1|}{2} - \frac{|Y|}{2} \text{,} \label{align::IS_H-Y}
    \end{align}
    because $V_0 \cup V_1 \subseteq H-Y$.
    
    Next, we try to find a lower bound for the size of a maximum independent set in $H$. 
    We will construct an independent set $I_H$ in $H$ and the size of this independent set is a lower bound for the size of a maximum independent set in $H$. First of all, we add all vertices from the independent set $V_1$ to $I_H$; this will prevent every vertex from $N_H(V_1)$ to be part of the independent set $I_H$. Now, we can extend the independent $V_1$ by an independent set in $H-N_H[V_1]$. First, observe that $N_H[V_1] \cap Y = \emptyset$, because $V_1 \subseteq (H-Y)$ and $N_H(Y) \cap V_1 = \emptyset$. From this follows that $H-N_H[V_1] = H - V_0 - V_1$, because $N(V_1)=V_0$. Instead of adding an independent set of $H-V_0-V_1$ to $I_H$, we add a maximum independent set $I_F$ of the forest $H- V_0 - V_1 - Z$ to $I_H$; such an independent set $I_F$ has size at least $\frac{1}{2} |H-V_0-V_1-Z|$. This leads to the following lower bound for $\alpha(H)$:
    \begin{align}
        \alpha(H) &\geq |I_H| = |V_1| + |I_F| \geq |V_1| + \frac{|H-V_0-V_1-Z|}{2} \nonumber \\ 
        &= |V_1| + \frac{|H-V_0-V_1|}{2} - \frac{|Z \setminus (V_0 \cup V_1)|}{2}
        \geq |V_1| + \frac{|H-V_0-V_1|}{2} - \frac{|Z|}{2} \label{align::IS_H}
    \end{align}
    
    Using the equation $\alpha(H)=\alpha(H-Y)+1$ together with the upper bound for $\alpha(H-Y)$ and the lower bound for $\alpha(H)$ leads to the requested upper bound for the size of $Y$:
    \begin{align*}
        |V_1| + \frac{|H-V_0-V_1|}{2} - \frac{|Z|}{2} \overset{(\ref{align::IS_H})}{\leq} \alpha(H) &= \alpha(H-Y) + 1 \overset{(\ref{align::IS_H-Y})}{\leq} |V_1| + \frac{|H-V_0-V_1|}{2} - \frac{|Y|}{2} + 1 \\
        \implies |Y| &\leq |Z|+2 \text{.}\qedhere
    \end{align*}
\end{proof}

We showed that every minimal blocking set in a $d$-quasi-forest has size at most $d+2$. To proof Theorem \ref{theorem::lower_bound_FVS} it remains to show that the bound is tight:

\begin{proof}[Proof of Theorem \ref{theorem::lower_bound_FVS}]
    We show the remaining part of Theorem \ref{theorem::lower_bound_FVS}, namely that the bound is tight.
    
    Consider the connected graph $H=K_{d+2}$. It holds that $H$ is a $d$-quasi-forest, because any $d$ vertices from $H$ are a feedback vertex set. It holds that the size of a maximum independent set in a clique is 1, hence $\alpha(H-Y')=1$ for all subsets $Y' \subsetneq V(H)$. Therefore, $Y=V(H)$ is the only, and hence a minimal, blocking set in $H$.
\end{proof}

Recall that Reduction Rule \ref{rule::delete_cc} considers the conflicts that a connected component $H$ of the $d$-quasi-forest $G-X$ has with subsets of $X$. So far, we only talked about the size of minimal blocking sets instead of the size of minimal subset of $X$ that leads to a conflict. 
Since every independent set $X_I \subseteq X$ that has a conflict with $H$, has some neighbors in this component, we know that these vertices are a blocking set of $H$. Using Lemma \ref{lemma::conflict} we can argue that only a subset of at most $d+2$ vertices (of the neighborhood of $X_I$ in $H$) is important. Like Jansen and Bodlaender~\cite{JansenB13} resp.\ Fomin and Str{\o}mme~\cite{FominS16} we show how a smaller subset of $V(H) \cap N(X_I)$ leads to a smaller subset of $X_I$ that has a conflict with the connected component $H$.

\begin{lemma} \label{lemma::conflict_size}
    Let $(G,X,k)$ be an instance of \IS parameterized by a modulator to a $d$-quasi-forest.
    Let $H$ be a connected component of $G-X$ and let $X_I \subseteq X$ be an independent set in $G$. If $\Conf{H}{X_I} > 0$, then there exists a set $X' \subseteq X_I$ of size at most $d+2$ such that $\Conf{H}{X'} > 0$.
\end{lemma}

\begin{proof}
    Let $Y=N(X_I) \cap V(H)$ be the neighborhood of $X_I$ in the connected component $H$; it holds that $\alpha(H) > \alpha(H-Y)$, because $0 < \Conf{H}{X_I}=\alpha(H)-\alpha(H-N(X_I))=\alpha(H)-\alpha(H-Y)$. Let $Y' \subseteq Y$ be a minimal blocking set. It follows from Lemma \ref{lemma::conflict} that $|Y'| \leq d+2$.
    
    We pick for every vertex $y \in Y'$ an arbitrary neighbor $x_y \in X_I$ in $X_I$. Let $X' = \{x_y \mid y \in Y'\} \subseteq X_I$. Clearly, it holds that $|X'| \leq d+2$ and that $Y' \subseteq N(X')$, which implies that $\Conf{H}{X'}= \alpha(H)-\alpha(H-N(X')) \geq \alpha(H)-\alpha(H-Y') > 0$; thus $X'$ has the desired properties.
\end{proof}

We showed that if a connected component $H$ of $G-X$ has a conflict with a subset $X' \subseteq X$ of the modulator, then there always exists a set $X'' \subseteq X'$ of size at most $d+2$ that has a conflict with the connected component $H$. 
Knowing this, we can show that Reduction Rule \ref{rule::delete_cc} is safe using Lemma \ref{lemma::conflict_size} as well as some observations that where already used in earlier work~\cite{FominS16, JansenB13}.

\begin{lemma} \label{lemma::rule1}
    Reduction Rule \ref{rule::delete_cc} is safe; let $(G,X,k)$ be the instance before applying Reduction Rule \ref{rule::delete_cc} and let $(G',X,k')$ be the reduced instance. Then there exists a solution for $(G,X,k)$ if and only if there exists a solution for $(G',X,k')$.
\end{lemma}

\begin{proof}
    Let $H$ be the connected component of $G-X$ that we delete by applying Reduction Rule \ref{rule::delete_cc}. For the forward direction of the proof, we assume that $(G,X,k)$ has a solution, thus there exists an independent set $I$ of size at least $k$ in $G$. Consider the set $I'=I \setminus V(H)$. Clearly, $I'$ is an independent set of $G'$ of size at least $|I|-|I \cap V(H)| \geq k-\alpha(H)=k'$, because $I \cap V(H)$ is an independent set in $H$. Therefore, $(G',X,k')$ has a solution, namely~$I'$.
    
    For the backward direction of the proof, we assume that $(G',X,k')$ has a solution, thus there exists an independent set $I'$ of size at least $k'$ in $G'$. 
    First, we will show that there always exists an independent set which can only contain an entire set $X' \subseteq X$ if this set induces strictly less than $|X|$ conflicts in $G'-X$. This was already shown by Jansen and Bodlaender~\cite{JansenB13}.
    
    Let $I'$ be an arbitrary independent set in $G'$ of size at least $k'$. Assume that there exists a set $X_I \subseteq X \cap I'$ such that $\Conf{G'-X}{X_I} \geq |X|$. We will construct an independent set $\widetilde{I}$ of the same size that contains no vertex from $X$ and therefore fulfills the desired property.
    Since $\Conf{G'-X}{X_I} \geq |X|$, we know that $\alpha(G'-X) - \alpha(G'-X-N(X_I)) \geq |X|$.
    The set $I' \setminus X$ is an independent set in $G'-X-N(X_I)$, because $I'$ is an independent set in $G'-X$ and $X_I \subseteq I'$ which implies that $N(X_I) \cap I' = \emptyset$. Thus, we know that $\alpha(G'-X-N(X_I)) \geq |I' \setminus X|$. Combining all this leads to 
    \begin{align*}
        |I'| \leq |I' \setminus X| + |X| \leq \alpha(G'-X-N(X_I)) + |X| \leq \alpha(G'-X) \text{.}
    \end{align*}
    Thus, every maximum independent set $\widetilde{I}$ of $G-X$ has at least the cardinality of $I'$ and fulfills the desired properties. 
    
    Now, we can assume that $I'$ is an independent set that can only contain an entire set $X' \subseteq X$ when this set induces strictly less than $|X|$ conflicts in $G'-X$.
    We will show that we can extend $I'$ to an independent set of size at least $k=k'+\alpha(H)$ by adding a maximum independent set of $H$ to $I'$. More precisely, we will show that $\alpha(H) = \alpha(H-N(X \cap I'))$; note that it suffices to show that $\alpha(H)\leq \alpha(H-N(X \cap I'))$. Observe that, if $\alpha(H) \leq \alpha(H-N(X \cap I'))$, then there exists a maximum independent set $I_H$ in $H$ that uses no vertex from the set $N(X \cap I')$, therefore $I= I' \cup I_H$ is an independent set in $G$ of size at least $k' + \alpha(H) = k$.
    
    Now, we assume for contradiction that $\alpha(H)>\alpha(H-N(X \cap I'))$ which implies that $\Conf{H}{X \cap I'}>0$. From Lemma \ref{lemma::conflict_size} it follows that there exists a set $\widetilde{X} \subseteq X \cap I'$ of size at most $d+2$ such that $\Conf{H}{\widetilde{X}} > 0$. Since we assumed that $I'$ can only contain subsets of $X$ that induce less than $|X|$ conflicts in $G'-X$ (and $\widetilde{X} \subseteq I' \cap X$) it holds that $\Conf{G'-X}{\widetilde{X}} < |X|$. But this contradicts the requirements of Reduction Rule \ref{rule::delete_cc}: $\widetilde{X} \subseteq X$ is an independent set of size at most $d+2$ with $\Conf{H}{\widetilde{X}}  > 0$ and $\Conf{G'-X}{\widetilde{X}} = \Conf{G-H-X}{\widetilde{X}} < |X|$. Thus, the assumption is wrong and we have $\alpha(H) \leq \alpha(H-N(X \cap I'))$.
\end{proof}

Recall that if we have an instance $(G,X,k)$ of \IS parameterized by a modulator to a $d$-quasi-forest where Reduction Rule \ref{rule::delete_cc} is not applicable then $G-X$ has at most $|X|^{d+3}$ connected components. To apply the kernelization for \IS parameterized by a modulator to a forest from Jansen and Bodlaender~\cite{JansenB13}, we have to add vertices from each connected component of the $d$-quasi-forest $G-X$ to the modulator $X$, getting a set $X'\supseteq X$, such that the connected components of $G-X'$ are trees. 

We know that every connected component of the $d$-quasi-forest $G-X$ has a feedback vertex set of size at most $d$, which we can find in polynomial time, since $d$ is a constant. Let $Z \subseteq V(G-X)$ be the union of these feedback vertex sets; it holds that $|Z|\leq d \cdot |X|^{d+3}$. Now, the instance $(G',X',k')$ with $G'=G$, $X'=X \cup Z$ and $k'=k$ is an instance of \IS parameterized by a modulator to feedback vertex set. Obviously, it holds that $(G,X,k)$ has a solution if and only if $(G',X',k')$ has a solution. Applying the following result of Jansen and Bodlaender \cite{JansenB13} will finish our kernelization.

\begin{proposition}[{\cite[Theorem 2]{JansenB13}}] \label{proposition::kernel}
    \IS parameterized by a modulator to a \FVS has a kernel with a cubic number of vertices: there is a polynomial-time algorithm that transforms an instance $(G,X,k)$ into an equivalent instance $(G',X',k')$ such that $|X'| \leq 2|X|$ and $|V(G')| \leq 2|X|+28|X|^2+56|X|^3$.
\end{proposition}

\begin{theorem} \label{theorem::quasi-forest}
    \IS parameterized by a modulator to a $d$-quasi-forest admits a kernel with $\Oh(d^3 |X|^{3d+9})$ vertices.
\end{theorem}

\begin{proof}
    Given an input instance $(G,X,k)$ of \IS parameterized by a modulator to a $d$-quasi-forest, we first apply Reduction Rule \ref{rule::delete_cc} exhaustively to obtain an equivalent instance $(\widetilde{G},X,\widetilde{k})$ of \IS parameterized by a modulator to a $d$-quasi-forest where the $d$-quasi-forest $\widetilde{G}-X$ has at most $|X|^{d+3}$ connected components (Lemma \ref{lemma::bound_cc}). The fact that the instances are equivalent follows from Lemma \ref{lemma::rule1}. Furthermore, we can compute the instance $(\widetilde{G},X,\widetilde{k})$ in polynomial time:
    Every application of Reduction Rule \ref{rule::delete_cc} deletes a connected component of $G-X$ and decreases the value of $k$ appropriately, hence we apply this rule at most $|V|$ times. 
    To apply Reduction Rule \ref{rule::delete_cc}, we have to find a connected component $H$ of $G-X$ that only has a conflict with an independent set $X_I \subseteq X$, when the set $X_I$ has a large conflict in the $d$-quasi-forest $G-X$. Thus, for every connected component $H$ of $G-X$ (at most $|V|$) and the $d$-quasi-forest $G-X$ we have to compute for at most $|X|^{d+2}$ sets $X_I \subseteq X$ the value $\Conf{H}{X_I}$ resp.\ $\Conf{G-X}{X_I}$. Since $d$ is a constant we can compute $\Conf{H}{X_I}$ and $\Conf{G-X}{X_I}$ for all sets $X_I \subseteq X$ in polynomial time and we can easily check whether a connected component $H$ fulfills the properties of Reduction Rule \ref{rule::delete_cc}. Summarized, instance $(\widetilde{G},X,\widetilde{k})$ is equivalent to instance $(G,X,k)$ and we can compute this instance in polynomial time. 
    
    Now, we add for each of the at most $|X|^{d+3}$ connected components of $\widetilde{G}-X$ a feedback vertex set of size at most $d$ to $X$; let $Z$ be the union of these feedback vertex sets. We add the vertex set $Z$ to the modulator $X$ to obtain an instance $(\widetilde{G},\widetilde{X}=X \cup Z,\widetilde{k})$ of \IS parameterized by a modulator to a forest. 
    
    The instances $(\widetilde{G},X,\widetilde{k})$ and $(\widetilde{G},\widetilde{X},\widetilde{k})$ are obviously equivalent. To prove that we can construct $(\widetilde{G},\widetilde{X},\widetilde{k})$ in polynomial time, we only have to show that we can find the set $Z$ in polynomial time; this holds, because we can find a feedback vertex of constant size in polynomial time. 

    Finally, we apply the kernelization algorithm of Proposition \ref{proposition::kernel} to the instance $(\widetilde{G},\widetilde{X},\widetilde{k})$ of \IS parameterized by a modulator to a forest and obtain an equivalent instance $(G',X',k')$ in polynomial time.
    
    So far, we know that we can compute instance $(G',X',k')$ of \IS parameterized by a modulator to a feedback vertex set, which is equivalent to the instance $(G,X,k)$ of \IS parameterized by a modulator to a $d$-quasi-forest, in polynomial time. It remains to bound the size of $V(G')$, $X'$ and $k'$. 
    
    We never increase the size of $k$, we only decrease the size of $k$ in Reduction Rule \ref{rule::delete_cc} and the application of Proposition \ref{proposition::kernel}, hence $k' \leq k$. Next, we bound the size of $X'$. We increase the cardinality of the set $X$ twice, once by adding the feedback vertex set $Z$ of the $d$-quasi-forest $\widetilde{G}-X$ to the modulator and once (by a factor of two) by applying Proposition \ref{proposition::kernel}. This leads to the following bound for the size of $X'$: $|X'| \leq 2 |\widetilde{X}| = 2 (|Z| + |X|) \leq 2 (d \cdot |X|^{d+3} + |X|)$.
   
   Finally, we have to bound the number of vertices in $G'$. It follows from applying Proposition \ref{proposition::kernel} to the instance $(\widetilde{G},\widetilde{X},\widetilde{k})$ of \IS parameterized by a modulator to a forest that the reduced instance $(G',X',k')$ has at most $2|\widetilde{X}| + 28 |\widetilde{X}|^2+56|\widetilde{X}|^3$ vertices. This leads to the desired bound for $V(G')$:
   \begin{align*}
   |V(G')| &\leq 2|\widetilde{X}| + 28 |\widetilde{X}|^2+56|\widetilde{X}|^3 \\
   &\leq 2(d \cdot |X|^{d+3} + |X|) + 28 (d \cdot |X|^{d+3} + |X|)^2+56(d \cdot |X|^{d+3} + |X|)^3 \\
   &\in \Oh(d^3 |X|^{3d+9}) \text{.}\qedhere
   \end{align*}
\end{proof}

\begin{corollary}
    \VC parameterized by a modulator to a $d$-quasi-forest admits a kernel with $\Oh(d^3 |X|^{3d+9})$ vertices.
\end{corollary}

\section{Two other graph classes with small blocking sets} \label{section::small_blocking_sets}

In this section we consider \VC parameterized by a modulator to a $d$-quasi-bipartite graph and by a modulator to a $d$-quasi-integral graph. As in the case of \VC parameterized by a modulator to a $d$-quasi-forest, we prove that the size of a minimal blocking set is bounded linearly in $d$ to reduce the number of connected components in the $d$-quasi-bipartite graph resp.\ the $d$-quasi-integral graph. Having only polynomial in the modulator many connected components we show that we can apply the randomized polynomial kernelizations for \VC parameterized by a modulator to a bipartite graph, resp.\ \VC above $\VCLP$.

The proof that there exists a kernelization for \VC parameterized by a modulator to a $d$-quasi-bipartite graph works just the same as the kernelization for \VC parameterized by a modulator to a $d$-quasi-forest, except for the last step. Here we apply the kernelization of \VC parameterized by a modulator to a bipartite~graph. 

\begin{corollary} \label{lemma::conflict_OCT}
    In a $d$-quasi-bipartite graph the size of a minimal blocking set has a tight upper bound of $d+2$.
\end{corollary}

\begin{proof}
    Let $H=(V,E)$ be a $d$-quasi-bipartite graph and let $Z$ be an odd cycle transversal in $H$ of size at most $d$.
    Like in the proof of Lemma \ref{lemma::conflict} we consider an optimum half integral solution $x$ to $\ISLP(H-Y)$ which fulfills the properties of Lemma \ref{lemma::LP_solution}. Let $V_i = \{v \in V(H-Y) \mid x_v=i\}$ for $i\in \{0, \frac{1}{2},1\}$.
    
    Note that the upper bound $\alpha(H-Y) \overset{(\ref{align::IS_H-Y})}{\leq} |V_1|+\frac{1}{2} |H-V_0-V_1|-\frac{1}{2}|Y|$ holds also in this case, because the value of an optimum half integral solution is always a valid upper bound.
    
    But also the lower bound $\alpha(H) \overset{(\ref{align::IS_H})}{\geq} |V_1| + \frac{1}{2} |H-V_0-V_1| - \frac{1}{2} |Z|$ holds in this case, because the independent set $V_1$ together with an independent set in the bipartite graph $H-V_0-V_1-Z$ is an independent set in $H$. Note that a maximum independent  set in a bipartite graph contains at least half the vertices. 
    Hence, we get $|Y| \leq |Z|+2$.
    
    It remains to show that the upper bound of $d+2$ vertices is tight. Note that the connected graph $H=K_{d+2}$ with $d=|Z|$ has only one blocking set, namely the set $V(H)$, which matches the upper bound (compare with Theorem \ref{theorem::lower_bound_FVS}).
\end{proof}

\begin{corollary}
    \VC parameterized by a modulator to a $d$-quasi-bipartite graph admits a randomized polynomial kernel.
\end{corollary}

\begin{proof}
    Let $(G,X,k)$ be an instance of \VC parameterized by a modulator to a $d$-quasi-bipartite graph. First, we transform it into an equivalent instance $(G,X,|V(G)|-k)$ of \IS parameterized by a modulator to a $d$-quasi-bipartite graph. 
    Just as for the kernelization for \IS parameterized by a modulator to a $d$-quasi-forest, we can obtain in polynomial time an equivalent instance $(\widetilde{G},\widetilde{X},|V(\widetilde{G})|-\widetilde{k})$ of \IS parameterized by a modulator to a bipartite graph, where the cardinality of $\widetilde{X}$ is at most $d \cdot |X|^{d+3} + |X|$. 
    
    We can apply Reduction Rule \ref{rule::delete_cc} and we can add at most $d \cdot |X|^{d+3}$ vertices to the modulator to obtain an instance of \IS parameterized by a modulator to a $d$-quasi-bipartite graph. The proof works analogously to Theorem \ref{theorem::quasi-forest}, because we can compute a maximum independent set in a $d$-quasi-bipartite graph in polynomial time (because $d$ is a constant) and because Lemma \ref{lemma::bound_cc} and Lemma \ref{lemma::rule1} still hold. The proofs of these lemmas use only the fact that there exists a minimal blocking set of size at most $d+2$ and general facts about maximum independent sets.

    Finally, transforming the instance $(\widetilde{G},\widetilde{X},|V(\widetilde{G})|-\widetilde{k})$ of \IS parameterized by a modulator to a bipartite graph into the equivalent instance $(\widetilde{G},\widetilde{X},\widetilde{k})$ of \VC parameterized by a modulator to a bipartite graph, we apply the randomized polynomial kernelization algorithm for \VC parameterized by a modulator to a bipartite graph \cite{KratschW12}. This leads to an instance $(G',X',k')$ in polynomial time. Furthermore, the size of the instance $(G',X',k')$ is bounded by a polynomial in the size of $\widetilde{X}$ and thus polynomially bounded in $|X|$, because $|\widetilde{X}| \leq d |X|^{d+3} + |X|$. 
\end{proof}

In contrast to $d$-quasi-forests and $d$-quasi-bipartite graphs, where every minimal blocking set is of size at most $d+2$, $d$-quasi-integral graphs have minimal blocking sets of size up to $2d+2$. 
Nevertheless, all proofs, to show that there exists a polynomial kernel, still work, because we only need the existence of a small blocking set. 

\begin{lemma} \label{lemma::conflict_LP}
    Let $H=(V,E)$ be a $d$-quasi-integral graph. Then it holds that a minimal blocking set $Y$ in the $d$-quasi-integral graph $H$ has size at most $2d+2$.
\end{lemma}

\begin{proof}
    Like in the proof of Lemma \ref{lemma::conflict} we consider an optimum half integral solution $x$ to $\ISLP(H-Y)$ which fulfills the properties of Lemma \ref{lemma::LP_solution}. Let $V_i = \{v \in V(H-Y) \mid x_v=i\}$ for $i\in \{0, \frac{1}{2},1\}$.
    
    Note that the upper bound $\alpha(H-Y) \overset{(\ref{align::IS_H-Y})}{\leq} |V_1|+\frac{1}{2} |H-V_0-V_1|-\frac{1}{2}|Y|$ also holds in this case, because the value of an optimum half integral solution is always a valid upper bound.
    
    In this case the lower bound for $\alpha(H)$ works slightly differently. Instead of constructing an independent set in $H$ we construct a feasible solution to $\ISLP(H)$. We first use the fact that $H$ is a $d$-quasi-integral graph, hence $\mathrm{vc}(H) \leq \VCLP(H) +d$, which is equivalent to $\alpha(H) \geq \ISLP(H) -d$, because $\alpha(H)= |H|-\mathrm{vc}(H)$ and $|H|-\VCLP(H) = \ISLP(H)$. Now, we construct a feasible solution $x'$ to $\ISLP(H)$. First, we assign every vertex $v$ in the independent set $V_1$ the value 1 and every vertex $w$ in $N_H(V_0)$ the value 0. Like in the proof of Lemma \ref{lemma::conflict}, it holds that $N_H[V_1] = H-V_0 - V_1$, because $N_H[V_1] \cap Y = \emptyset$. Finally, we assign the value $\frac{1}{2}$ to every vertex in $H- V_0-V_1$. Obviously, $x'$ is a feasible solution to $\ISLP(H)$. This leads to the following lower bound for $\alpha(H)$:
    \begin{align}
        \alpha(H) \geq \ISLP(H) -d \geq |V_1| + \ISLP(H-V_0-V_1) - d \geq |V_1| + \frac{|H-V_0-V_1|}{2} - d \label{align::bound_H}
    \end{align}
    
    Again, using the equation $\alpha(H)=\alpha(H-Y)+1$ together with the upper bound for $\alpha(H-Y)$ and the lower bound for $\alpha(H)$ leads to the requested bound for the size of $Y$:
    \begin{align*}
        |V_1| + \frac{|H-V_0-V_1|}{2} - d \overset{(\ref{align::bound_H})}{\leq} \alpha(H) &= \alpha(H-Y) + 1 \overset{(\ref{align::IS_H-Y})}{\leq} |V_1| + \frac{|H-V_0-V_1|}{2} - \frac{|Y|}{2} + 1 \\
        \implies |Y| &\leq 2d+2 \text{.}\qedhere
    \end{align*}
\end{proof}

\begin{theorem}
    In a $d$-quasi-integral graph the size of a minimal blocking set has a tight upper bound of $2d+2$.
\end{theorem}

\begin{proof}
The upper bound of $2d+2$ follows from Lemma~\ref{lemma::conflict_LP}. For tightness consider the connected graph $H=K_{2d+2}$. It holds that $\mathrm{vc}(H) \leq \VCLP(H) + d$, because $\mathrm{vc}(H) = 2d+1$ and $\VCLP(H) = \frac{1}{2} |H|=d+1$. Furthermore, the size of a maximum independent set in a clique is 1, thus $\alpha(H-Y')=1$ for all subsets $Y' \subsetneq V(H)$. Therefore, $Y=V(H)$ is the only, and hence a minimal, blocking set in $H$.
\end{proof}

\begin{theorem}
    \VC parameterized by a modulator to a $d$-quasi-integral graph admits a randomized polynomial kernel.
\end{theorem}

\begin{proof}
    Let $(G,X,k)$ be an instance of \VC parameterized by a modulator to a $d$-quasi-integral graph. First, we transform it into an equivalent instance $(G,X,|V(G)|-k)$ of \IS parameterized by a modulator to a $d$-quasi-integral graph. 
    Just as for the kernelization for \IS parameterized by a modulator to a $d$-quasi-forest, we can obtain in polynomial time an equivalent instance $(\widetilde{G},X,|\widetilde{G}|-\widetilde{k})$ of \IS parameterized by a modulator to a $d$-quasi-integral graph by applying Reduction Rule \ref{rule::delete_cc} exhaustively. Now, $\widetilde{G}-X$ has at most $|X|^{2d+3}$ connected components. 
    
    The proof works analogously to Theorem \ref{theorem::quasi-forest}, because we can compute a maximum independent set in a $d$-quasi-integral graph in polynomial time (because $d$ is a constant) and because Lemma \ref{lemma::bound_cc} and Lemma \ref{lemma::rule1} still hold. The proofs of these Lemmas use only the fact that there exists a minimal blocking set of size at most $d+2$ and general facts about maximum independent sets. (Since we have a minimal blocking set of size $2d+2$ we have to replace every $d$ by $2d$).
    Now, we transform the instance $(\widetilde{G},X,|V(\widetilde{G})|-\widetilde{k})$ of \IS parameterized by a modulator to a $d$-quasi-integral graph into the equivalent instance $(\widetilde{G},X,\widetilde{k})$ of \VC parameterized by a modulator to a $d$-quasi-integral graph.
    
    Note that a vertex cover in $\widetilde{G} - X$ together with the set $X$ is a vertex cover of $\widetilde{G}$; the size of this vertex cover is $\mathrm{vc}(\widetilde{G}-X) + |X|$. We can assume that $\widetilde{k}$ is strictly smaller than the size of this vertex cover; otherwise the set $X \cup X_{\widetilde{G}-X}$ is a vertex cover of $G$ of size at most $\widetilde{k}$ that we can compute in polynomial time, where $X_{\widetilde{G}-X}$ is a minimum vertex cover in $\widetilde{G}-X$. Thus, in the following we assume that $\mathrm{vc}(\widetilde{G}-X) + |X| > \widetilde{k}$.
    
    Finally, we apply the kernelization algorithm for \VC above $\VCLP$ to the instance $(\widetilde{G},\widetilde{k})$ and obtain an instance $(G',k')$ in polynomial time. Note that we can bound the parameter $\widetilde{k} - \VCLP(\widetilde{G})$ by a polynomial in the size of $X$ as follows:
    \begin{align*}
        \widetilde{k} - \VCLP(\widetilde{G}) 
        &\leq \widetilde{k} - \VCLP(\widetilde{G}-X) \\
        &= \widetilde{k} - \sum_{H \text{ c.c of } \widetilde{G}-X}  \VCLP(H) \\
        &\leq \widetilde{k} - \sum_{H \text{ c.c of } \widetilde{G}-X} (\mathrm{vc}(H) - d) \\
        &= \widetilde{k} + |X|^{2d+3} d -  \mathrm{vc}(\widetilde{G}-X) \\
        &\leq |X| + |X|^{2d+3} d
    \end{align*}
    
    Since $(G',k')$ is polynomially bounded in the size of $\widetilde{k}-\VCLP(\widetilde{G})$, which is bounded by a polynomial in the size of $|X|$, we know that the instance $(G',X'=V(G'),k')$ is an equivalent instance of \VC parameterized by a modulator to a $d$-quasi-integral graph.
\end{proof}

\section{Conclusion}

Starting from the work of Fomin and Str{\o}mme~\cite{FominS16} we have presented new results for polynomial kernels for \VC subject to structural parameters. Our results for modulators to $d$-quasi-forests show that bounds on the feedback vertex set size are more meaningful for kernelization than the treewidth of $G-X$ (recalling that there is a lower bound for treewidth of $G-X$ being at most two). By extending our kernelization to work for modulators to ($d$-quasi-bipartite and) $d$-quasi-integral graphs, we have encompassed existing kernelizations for parameterization by distance to forests~\cite{JansenB13}, distance to max degree two~\cite{MajumdarRS15} (both previously subsumed by), distance to pseudoforests~\cite{FominS16}, and parameterization above fractional optimum~\cite{KratschW12}. It would be interesting whether there is a single positive result that encompasses all parameterizations with polynomial kernels.

To obtain our results we have established tight bounds for the size of minimal blocking sets in $d$-quasi-forests, $d$-quasi-bipartite graphs, and $d$-quasi-integral graphs. Tightness comes from the fact that cliques of size $d+2$ respectively $2d+2$ are contained in these classes. The presence of these cliques also implies lower bounds ruling out kernels of size $\Oh(|X|^{r-\epsilon})$, assuming \ncontainment, when $r=r(d)$ is the maximum size of minimal blocking sets as a consequence of a lower bound by Majumdar et al.~\cite{MajumdarRS15}. It would be interesting whether there are matching upper bounds for kernelization, e.g., whether the kernelization of Jansen and Bodlaender~\cite{JansenB13} for modulators to forests can be improved to size $\Oh(|X|^2)$.

\bibliography{lit}

\end{document}